\documentclass[10pt,a4paper]{article}
\usepackage[T1]{fontenc}
\usepackage[utf8]{inputenc}
\usepackage[english]{babel}
\usepackage{a4wide}
\usepackage{authblk}
\usepackage{amsfonts}
\usepackage{xspace}
\usepackage{nameref,hyperref}
\usepackage{stmaryrd}
\usepackage{mathrsfs}
\usepackage{float}
\usepackage{graphicx}
\usepackage{url}
\usepackage{amsmath,amsthm,amssymb,color}
\usepackage{multicol}
\usepackage{caption}
\usepackage{subcaption}
\usepackage{numprint}
\usepackage{cleveref}
\usepackage{tikz}
\usepackage{physics}
\usepackage{pdflscape}
\usepackage{array,multirow}
\usepackage{relsize}
\usepackage{markdown}
\usepackage{empheq}
\usepackage{verbatim}
\usepackage{multirow}
\usepackage{algorithm}
\usepackage{algpseudocode}
\usepackage[normalem]{ulem} 
\usepackage{xcolor,colortbl}
\usepackage[backend=bibtex,firstinits,style=alphabetic,url=false,maxcitenames=3,maxnames=8]{biblatex}

\DeclareNameAlias{labelname}{given-inits}

\algnewcommand\algorithmicforeach{\textbf{for each}}
\algdef{S}[FOR]{ForEach}[1]{\algorithmicforeach\ #1\ \algorithmicdo}
\algnewcommand{\IIf}[1]{\State\algorithmicif\ #1\ \algorithmicthen}
\algnewcommand{\ElseIIf}[1]{\State \textbf{else} \algorithmicif\ #1\ \algorithmicthen}
\algnewcommand{\ITE}[3]{\State\algorithmicif\ #1\ \algorithmicthen\ #2\ \algorithmicelse\ #3\ \algorithmicend\ \algorithmicif}
\algnewcommand{\EndIIf}{\unskip\ \algorithmicend\ \algorithmicif}
\algnewcommand{\IFor}[2]{\State\algorithmicfor\ #1\ \algorithmicdo \ #2\ \algorithmicend\ \algorithmicfor}

\usetikzlibrary{calc} 
\usetikzlibrary{backgrounds}

\usepackage{mathtools}

\newcommand{\qbin}[2]{\begin{bmatrix}{#1}\\ {#2}\end{bmatrix}_q}

\renewcommand{\innerproduct}[2]{\left\langle #1\mid#2 \right\rangle}
\newcommand{\ctimes}{\cdot}

\theoremstyle{plain}
\newtheorem{theorem}{Theorem}
\newtheorem{corollary}[theorem]{Corollary}

\theoremstyle{definition}
\newtheorem{definition}{Definition}

\colorlet{cRed}{red!100!}

\begin{document}

\title{Automated search for highly contextual Kochen-Specker proofs}

\author{Axel Muller$^1$ and Metod Saniga$^2$}

\affil{$^1$Inria, ENS de Lyon, UCBL, LIP, 69342, Lyon Cedex 07, France\\
$^2$Astronomical Institute, Slovak Academy of Sciences SK-05960 Tatranská Lomnica, Slovak Republic}

\date{}

\maketitle

\begin{abstract}
Observable-based Kochen-Specker proofs are configurations of multi-qubit Pauli observables grouped into contexts whose products are plus or minus the indentity. Their robustness as state-independent contextuality tests can be measured by the tolerated error per context $\varepsilon = 2d/|H|$, where $d$ is the contextuality degree and $|H|$ the number of contexts. Since the degree depends only on an underlying abstract structure called the hypergram i.e. the pair formed by the context hypergraph and the anticommutation graph, the search for highly contextual proofs can be carried out on them instead, with no reference to qubits or to any particular Pauli labeling. We further exploit this by enumerating anticommutation graphs first, and then by associating to each graph $G$ the single hypergram carrying its entire hypergraph support $HS(G)$, so that exactly one candidate is examined per graph. Applied to the House of Graphs database and to censuses of vertex-transitive graphs on at most 24 vertices, this pipeline recovers the Peres-Mermin square, the doily and the Mermin pentagram, and yields configurations reaching $\varepsilon = 0.707$, against $0.424$ for the previous published record. The best configurations are predominantly those stemming from line graphs and unions of graphs; we explain the former by showing that every perfect matching of a graph is a context of its line graph, which exhibits the Peres-Mermin square and the doily as the first members of two infinite families. We close with finite geometric descriptions of the most striking configurations inside symplectic polar spaces, in terms of ovoids, hyperbolic quadrics and Fano planes.
\end{abstract}

\section{Introduction}
\label{seCintroduction}

In classical physical theories, the measured value of a physical quantity
does not depend on the other quantities measured alongside it, which form
its \emph{context}. Quantum theory violates this independence: the
Kochen-Specker theorem predicts experiments whose outcomes depend on which
compatible measurements are performed with them, a phenomenon called
\emph{quantum contextuality}~(see, e.\,g., \cite{bcgkl} for a review of the topic), now understood as a resource
for quantum computation. It has been formalised in several ways, through
binary~\cite{Cleve2014} and linear~\cite{CLS17} constraint systems,
through graphs~\cite{CSW14} and hypergraphs~\cite{AFL+15} whose vertices
are events and whose edges collect mutually exclusive ones, through
homotopy~\cite{OR20} and through Lie algebras~\cite{ACE+25}.

We work with \emph{observable-based} contextuality proofs, whose
measurements are multi-qubit Pauli observables. Such a proof is described
by a hypergraph whose vertices are the observables and whose hyperedges,
the \emph{contexts}, consist of mutually commuting observables whose product is
$\pm I^{\otimes n}$, where the sign distinguishes \emph{positive} from
\emph{negative} contexts. (These hypergraphs are not those
of~\cite{CSW14,AFL+15}, whose vertices are events rather than
observables.) These proofs are \emph{state-independent}, since they
witness contextuality for any initial state, and for a small number of
qubits they are \emph{testable} on existing quantum
computers~\cite{KZGKGCBR09,Hol21}. How contextual such a proof is, and
hence how much experimental noise the corresponding test tolerates, is
measured by its \emph{contextuality degree}~\cite{DHGMS22} i.e. the least
number of context signs that no assignment of measurement values to the
observables can reproduce. Finding proofs of
high degree is the objective of this article.

The proofs of this kind that have been searched for so far are mostly about
\emph{magic sets}~\cite{HS17}, whose observables lie in an even number of
contexts, whose negative contexts are odd in number, and in which two
observables share at most one context. These three conditions make
contextuality an immediate logical consequence, at the cost of confining
the search to a narrow region: the Peres-Mermin square~\cite{mermin,peres},
with its nine two-qubit observables and six ternary contexts, is the
smallest example. Configurations violating them are nonetheless
contextual, as multi-qubit doilies~\cite{MSGDH22} as well as bigger 
configurations~\cite{MSGDH24} show.

In previous work~\cite{mg25} the notion of \emph{hypergram} was introduced, the pair
formed by the context hypergraph and the graph of anticommutations of a
proof. It was shown that the contextuality degree is an invariant of this
structure alone, meaning that every Pauli labeling of a given hypergram, whatever its
number of qubits, has the same degree. 
The final piece of the puzzle comes with the idea of associating to an anticommutation graph the single hypergram carrying all of its admissible contexts, this was proposed by one of us in~\cite[Section 7.5]{muller-thesis}, where a preliminary and non-systematic search along these lines produced one configuration reaching an error per context $\varepsilon = 0.5$, identified here as $L(C(10,\{1,3\}))$ and beating the former published record from~\cite{TLC22} of $0.424$, although the latter one explicitely had its search restricted to magic sets. That material has not appeared outside the thesis and has not been peer-reviewed, whereas the present article develops the idea into a complete pipeline and takes it considerably further.

The scope of this article is therefore the following. We take quantum configurations as our object of study, and as our quantity to maximize the tolerated error per context $\varepsilon$, which is an experimentally measurable quantity proportional to the ratio between the contextuality degree and the total number of contexts. The search is organised around further restricting the search space from that of all hypergrams, to a mere anticommutation graph search, as a basis for determining which contexts are admissible, and then decide to keep all of them, so that exactly one hypergram is examined per graph. The problem then becomes the one of finding graphs whose contexts yield a high $\varepsilon$, and we attack it by running this pipeline over existing graph databases. Our goal is first to exhibit configurations whose ratio exceeds those of well-known configurations such as the Mermin pentagram and of the doily, and then to identify which structural features of the anticommutation graph accompany a high ratio. 

\Cref{seCdefs} fixes the notation and the background on the Pauli group, symplectic polar spaces, graphs and hypergrams, before introducing the new concept of hypergraph support; \Cref{secMethod} then justifies the choices made to simplify and narrow the research process, before \Cref{secResults} reports what the search returns, together with a theorem accounting for the recurrence of line graphs among the best configurations by identifying their contexts with the perfect matchings of the base graph, before commenting on the union of graphs yielding the best results. \Cref{secGeo} then steps back to describe the most striking of these configurations as geometric objects living in symplectic polar spaces, before \Cref{seCdiscussion} closes on the limits of the approach and the perspectives it opens.

\section{Definitions and notations}
\label{seCdefs}

\Cref{seCbackground,symplSec} recall the minimal background on the multi-qubit Pauli group and on its identification with the symplectic polar space $W_n$, along with the ovoids, Fano planes and hyperbolic quadrics that will be needed only in \Cref{secGeo}; \Cref{graphSec} does the same for the graph-theoretic vocabulary — perfect matchings, line graphs, transitivity, strong regularity — in which the outcome of the search is later phrased. \Cref{ccsSec} then recalls from~\cite{mg25} the abstract structure on which the whole article rests, the hypergram, makes explicit the notion of Pauli assignments and quantum configurations and introduces alongside it a new notion, the hypergraph support $HS(G)$ of an anticommutation graph. Finally~\Cref{contextualitySec} recalls the correspondence between the notions of contextuality degree~\cite{DHGMS22,MSGDH24} and noncontextual bound~\cite{Cab10,TLC22v2} before presenting the tolerated error per context $\varepsilon$ the single quantity maximised in the rest of the article.

\subsection{Multi-qubit Pauli group}
\label{seCbackground}

Let
\begin{equation*}
X = \left(
\begin{array}{rr}
0 & 1 \\
1 & 0 \\
\end{array}
\right),~~
Y = \left(
\begin{array}{rr}
0 & -\text{i} \\
\text{i} & 0 \\
\end{array}
\right)~~{\rm and}~~
Z = \left(
\begin{array}{rr}
1 & 0 \\
0 & -1 \\
\end{array}
\right)
\label{paulis}
\end{equation*}
be the Pauli matrices, $I$ the $2\times 2$ identity matrix, `$\otimes$' denote
the tensor product of matrices and $I^{\otimes n}$ denote the $n$-fold tensor
$I\otimes I \otimes\ldots\otimes I$ of the identity. A local \emph{$n$-qubit
(Pauli) observable} is a tensor product $A_1 \otimes A_2 \otimes \cdots \otimes
A_n$ with $A_i \in \{I,X,Y,Z\}$, usually denoted $A_1 A_2 \cdots A_n$, by
omitting the symbol $\otimes$ for the tensor product. Let '$\ctimes$' denote the
matrix product and $M^2$ denote $M \ctimes M$. It is easy to check that $X^2 =
Y^2 = Z^2 = I$, $X \ctimes Y = \text{i}Z = -Y \ctimes X$, $Y \ctimes Z =
\text{i}X = -Z \ctimes Y$, and $Z \ctimes X = \text{i}Y = -X \ctimes Z$. The
$n$-qubit observables with the multiplicative factors $\pm 1$ and $\pm
\text{i}$, called \emph{phase}, form the (\emph{generalized}) ($n$\emph{-qubit})
\emph{Pauli group} $\mathcal{P}^{\otimes n} =
(\{1,-1,\text{i},-\text{i}\}\times\{I,X,Y,Z\}^{\otimes n},\ctimes)$.

\subsection{Connection with symplectic polar spaces}
\label{symplSec}

Let $a$ and $b$ be two elements of the two-element field $\mathbb{F}_2 =
\{0,1\}$. Their sum, denoted $a+b$, and their product, denoted $ab$,
respectively correspond to the logical operations of exclusive disjunction and
conjunction, when $0$ encodes ``false'' and $1$ encodes ``true''.

The $2n$-dimensional vector space $\mathbb{F}_2^{2n}$ over $\mathbb{F}_2$ has
vector subspaces for each dimension $0 \leq k \leq 2n$. A subspace is
\emph{totally isotropic} if any two vectors $x$ and $y$ in it are mutually
orthogonal ($\innerproduct{x}{y} = 0$), for the symplectic form
$\innerproduct{.}{.}$ defined by
\begin{equation} 
\innerproduct{x}{y} = x_1y_2 + x_2y_1 + x_3y_4 + x_4y_3 + \dots 
              + x_{2n-1} y_{2n} + x_{2n} y_{2n-1}.
\label{symplf}
\end{equation}

The totally isotropic subspaces of $\mathbb{F}_2^{2n}$, without their zero
vector, form the \emph{symplectic (polar) space} $\mathcal{W}(2n-1,2)$
 of projective dimension $2n-1$. This name, in which $2$
is the order of the field $\mathbb{F}_2$, is hereafter shortened as $W_n$. In
other words, a (totally isotropic) subspace of $W_n$ of (projective) dimension
$k$, with $1 \leq k \leq n-1$, is a totally isotropic vector subspace of
$\mathbb{F}_2^{2n}$ of dimension $k+1$ without its $0$.
The symplectic polar space of rank two, $W_2$, is often dubbed as the {\it doily}.

\begin{definition}[Fano plane]\label{fanoDef}
The \emph{Fano plane} is the projective plane $\mathrm{PG}(2,2)$ over
$\mathbb{F}_2$, it has seven points and seven lines,
with three points on each line, three lines through each point, and a
unique line through any two distinct points.
\end{definition}

The $4^n-1$ phase-free $n$-qubit observables $A_1 \cdots A_j \cdots A_n$ in
$\mathcal{P}^{\otimes n}$ other than the identity $I^{\otimes n}$ are
bijectively identified with the $4^n-1$ vectors
$(x_1,x_2,\ldots,x_{2j-1},x_{2j},\ldots,x_{2n-1},x_{2n})$ which are the points
of $W_n$, by the extension $\psi:
\{I,X,Y,Z\}^{\otimes n}~\rightarrow~\mathbb{F}_2^{2n}$ of the encoding bijection
$\psi : \{I,X,Y,Z\} \rightarrow \mathbb{F}_2^{2}$ defined by
\begin{equation}
\psi(I) = (0,0),~\psi(X) = (0,1),~\psi(Y) = (1,1) {\rm~and}~\psi(Z) = (1,0).
\label{paulipts}
\end{equation}
This extension is defined by $\psi(A_1 \cdots A_j \cdots A_n) =
(x_1,x_2,\ldots,x_{2j-1},x_{2j},\ldots,x_{2n-1},x_{2n})$ with $\psi(A_j) =
(x_{2j-1},x_{2j})$ for $1~{\leq}~j~{\leq}~n$.

With the symplectic form defined by~(\ref{symplf}), two commuting observables
are represented by two orthogonal vectors.

\begin{definition}[Ovoid]\label{ovoidDef}
Let us call a maximal set of mutually anticommuting observables of $W_n$ an ovoid; each ovoid of $W_n$ has $2n+1$ elements. For $n=3$, an ovoid is also known as a Conwell heptad (see, e.\.g,~\cite{SdHG21}). 
\end{definition}

It is known that  for $W_n$ the number of its $k$-dimensional subspaces is
given by 
\begin{equation}
  \qbin{n}{k+1}  \prod_{i=1}^{k+1} (q^{n+1-i} +1),
  \label{sub-sympl}
\end{equation}
where 
\begin{equation}
\qbin{n}{k} = \prod_{i=1}^{k} \frac{q^{n-k+i} -1}{q^i -1} = \frac{(q^n-1) \dots (q^{n-k+1} -1)}{(q^k-1) \dots (q-1)}
\label{gauss}
\end{equation}
is the Gaussian (binomial) coefficient.

A {\it hyperbolic} quadric of $W_n$, $\mathcal{Q}^{+}(2n - 1, 2)$, for $n \geq 1$, is defined by the standard canonical equation:

   \[
   x_1x_{n+1}+x_2x_{n+2}+\cdots+x_{n}x_{2n}=0.
   \]
\noindent
   Each $\mathcal{Q}^{+}(2n - 1, 2)$ is endowed with $(2^{n-1} + 1)(2^{n} -1)$ points and there are $(2^{n-1} + 1)(2^{n} -1) + 1$ copies of them in $W_n$.
	An $n$-qubit observable $\mathcal{O}$ is either symmetric ($\mathcal{O}^{{\rm T}} = \mathcal{O}$), or skew-symmetric ($\mathcal{O}^{{\rm T}} = - \mathcal{O}$).
	Given a symmetric observable $\mathcal{O}$,  the set of symmetric observables commuting with $\mathcal{O}$ together with the set of skew-symmetric observables not
commuting with $\mathcal{O}$ lie on a certain $\mathcal{Q}^{+}(2n - 1, 2)$;  we  will call this associated observable the {\it index} of the quadric.

\subsection{Graph-related definitions}
\label{graphSec}
\begin{definition}[Perfect Matching]
Let \( G = (V, E) \) be a simple undirected graph. 
A \emph{perfect matching} in \( G \) is a subset \( M \subseteq E \) such that every vertex \( v \in V \) is incident to exactly one edge in \( M \). 
Equivalently, \( M \) is a matching that covers all vertices of \( G \).
\end{definition}

\begin{definition}[Line graph]\label{lineGraphDef}
The \emph{line graph} $L(G)$ of a simple graph $G = (V,E)$ is the graph
whose vertices are the edges in $E$, two of them being adjacent when
they share an endpoint.
\end{definition}

\begin{definition}[Rook's graph]\label{rookDef}
The \emph{$m \times m'$ rook's graph} is the Cartesian product
$K_m \square K_{m'}$: its vertices are the cells of an $m \times m'$ grid,
two of them being adjacent when they lie in the same row or in the same
column.
\end{definition}

It is worth mentioning that $K_m \square K_{m'} = L(K_{m,m'})$

\begin{definition}[Vertex- and edge-transitive graphs]\label{transitiveDef}
Let $(V,E)$ be a simple graph and $\mathrm{Aut}(G)$ its automorphism
group, that is, the group of permutations of $V$ preserving adjacency.
The graph is \emph{vertex-transitive} if $\mathrm{Aut}(G)$ acts
transitively on $V$, that is, if for any two vertices $u$ and $v$ there
is an automorphism sending $u$ to $v$. It is \emph{edge-transitive} if
$\mathrm{Aut}(G)$ acts transitively on $E$, that is, if any edge can be
mapped to any other edge by an automorphism.
\end{definition}

\begin{definition}[Strongly regular graph]
A graph is \emph{strongly regular} with
parameters $(\nu,\kappa,\lambda,\mu)$ if it has $\nu$ vertices, is
$\kappa$-regular, and any two vertices have exactly $\lambda$ common
neighbours when adjacent and exactly $\mu$ when non-adjacent.
\end{definition}

\begin{definition}[Cayley graph]
A \emph{Cayley graph} on a group $\Lambda$ has the elements of $\Lambda$
as vertices, two of them being adjacent when they differ by an element
of a fixed subset $S \subseteq \Lambda$ avoiding the identity and closed
under inverses.
\end{definition}

\begin{definition}[Circulant graph]
A \emph{circulant graph} $C(\nu,l)$ has $\nu$ vertices,
labelled by the residues modulo $\nu$, in which vertex $i$ is joined to
$i+j$ and to $i-j$ for every $j$ in a given list $l$ of
offsets.
\end{definition}

\subsection{Hypergrams and assignments}\label{ccsSec}

A \emph{hypergraph} $\mathcal{H} = (V,H)$ consists of a finite set $V$ of vertices $v$ and a set $H$ of hyperedges $h$ with $h \subseteq V$. Following~\cite{GR01,MWcirculant}, a simple graph is \emph{reduced} (or twin-free) if it has no isolated vertices and no distinct vertices with identical neighborhoods. Let $\text{cplt}(H)$ denote the complement graph on $V$, whose edges connect distinct pairs of vertices that do not share any hyperedge in $H$.

\begin{definition}
    
\label{hypergramDef}A \emph{hypergram}~\cite{mg25} is a triple $(V,H,G)$ where $(V,H)$ is a hypergraph without isolated vertices or empty hyperedges, and $(V,G)$ is a simple reduced \emph{anticommutation graph} satisfying $G \subseteq \text{cplt}(H)$. Vertices $i$ and $j$ \emph{commute} if $\{i,j\} \not\in G$.
\end{definition}

When $G = \text{cplt}(H)$, the hypergram $(V,H,G)$ is identified directly with its hypergraph $(V,H)$, as studied in settings like magic sets~\cite{TLC22}.

A quantum configuration~\cite{MSGDH24,mg25} is a pair $(O,C)$ in which $O$ is a non-empty finite set of $2^n$-dimensional Hermitian observables and $C$ a finite set of non-empty subsets of $O$, called contexts, such that
(i) $a^2=I^{\otimes n}$ for every $a\in O$, so that each observable has eigenvalues in $\{-1,+1\}$;
(ii) any two observables lying in a common context commute;
(iii) the product of the observables of each context is $I^{\otimes n}$ (positive context) or 
$-I^{\otimes n}$ (negative context).

An \emph{$n$-qubit Pauli assignment} of a hypergram $(V,H,G)$ is an injective map $\alpha : V \rightarrow \{I,X,Y,Z\}^{\otimes n} \setminus \{I^{\otimes n}\}$ assigning distinct non-identity observables such that:

\begin{enumerate}[(i)]
    \item \emph{Commutation condition}: $\alpha(v_1)$ and $\alpha(v_2)$ anticommute if and only if $\{v_1, v_2\} \in G$;
    \item \emph{Product condition}: $\prod_{v \in h} \alpha(v) = \pm I^{\otimes n}$ for every hyperedge $h \in H$.
\end{enumerate}

This allows us to define the quantum configuration associated to a Pauli assignment $\alpha$ of a hypergram $(V, H, G)$
as the pair $(O, C)$ such that $O = \alpha(V)$ and $C = \alpha(H)$. 
The notion of Pauli assignment connects to quantum satisfying assignments in binary constraint systems~\cite{Cleve2014}, operator solutions to binary linear systems~\cite{CLS17}, quantum realizations of signed arrangements~\cite{Arkhipov2012,OR20}, measurement-outcome representations~\cite{CSW14,AFL+15}, and Pauli-based assignments~\cite{TLC22v2}.

\begin{figure}[h!]
    \centering
    \begin{subfigure}[b]{0.49\textwidth}
        \centering
          \begin{tikzpicture}[scale=2, 
      every node/.style={scale=0.9,minimum size=1.0cm,circle,draw=black},
      ]
      \node (o1) at (0,2) {$v_1$};
      \node (o2) at (1,2) {$v_2$};
      \node (o3) at (2,2) {$v_3$};
      \node (o4) at (0,1) {$v_4$};
      \node (o5) at (1,1) {$v_5$};
      \node (o6) at (2,1) {$v_6$};
      \node (o7) at (0,0) {$v_7$};
      \node (o8) at (1,0) {$v_8$};
      \node (o9) at (2,0) {$v_9$};
      \draw[draw=gray] (o1) -- (o5);
      \draw[draw=gray] (o1) -- (o6);
      \draw[draw=gray] (o1) -- (o8);
      \draw[draw=gray] (o2) -- (o4);
      \draw[draw=gray] (o2) -- (o6);
      \draw[draw=gray] (o2) -- (o7);
      \draw[draw=gray] (o2) -- (o9);
      \draw[draw=gray] (o3) -- (o4);
      \draw[draw=gray] (o3) -- (o5);
      \draw[draw=gray] (o3) -- (o8);
      \draw[draw=gray] (o4) -- (o8);
      \draw[draw=gray] (o4) -- (o9);
      \draw[draw=gray] (o5) -- (o7);
      \draw[draw=gray] (o5) -- (o9);
      \draw[draw=gray] (o6) -- (o7);
      \draw[draw=gray] (o6) -- (o8);

      \draw[draw=gray,bend left=17]  (o1) to (o9);
      \draw[draw=gray,bend right=17] (o3) to (o7);
    \end{tikzpicture}
        \caption{Anticommutation graph $G_1$ of The Peres-Mermin square, the rook graph $K_3 \square K_3$, the grid structure of the edges is not visible in this layout, to better explain the Peres-Mermin assignment.}
        \label{fig:sub1}
    \end{subfigure}
    \hfill
    \begin{subfigure}[b]{0.49\textwidth}
        \centering
\begin{tikzpicture}[scale=2, 
every node/.style={scale=0.9,draw=black,circle,minimum size=10mm,}
]
\node (XY) at (0,0) {$v_7$};
\node (YX) at (1,0) {$v_8$};
\node (ZZ) at (2,0) {$v_9$};
\node (IY) at (0,1) {$v_4$};
\node (YI) at (1,1) {$v_5$};
\node (YY) at (2,1) {$v_6$};
\node (XI) at (0,2) {$v_1$};
\node (IX) at (1,2) {$v_2$};
\node (XX) at (2,2) {$v_3$};
\draw (XY) -- (IY) -- (XI) ;
\draw (YX) -- (YI) -- (IX) ;
\draw (ZZ) -- (YY) -- (XX) ;
\draw (XY) -- (YX) -- (ZZ) ;
\draw (IY) -- (YI) -- (YY) ;
\draw (XI) -- (IX) -- (XX) ;
\end{tikzpicture}
        \caption{Context hypergraph of the whole hypergraph support $HS(G_1)$.\\~\\~}
        \label{fig:sub2}
    \end{subfigure}
    \\
    \begin{subfigure}[b]{0.3\textwidth}
        \centering
\begin{tikzpicture}[scale=2, 
every node/.style={scale=1.0,draw=black,circle}
]
\node (XY) at (0,0) {$XY$};
\node (YX) at (1,0) {$YX$};
\node (ZZ) at (2,0) {$ZZ$};
\node (IY) at (0,1) {$IY$};
\node (YI) at (1,1) {$YI$};
\node (YY) at (2,1) {$YY$};
\node (XI) at (0,2) {$XI$};
\node (IX) at (1,2) {$IX$};
\node (XX) at (2,2) {$XX$};
\draw (XY) -- (IY) -- (XI) ;
\draw (YX) -- (YI) -- (IX) ;
\draw [double distance = 2pt,red] (ZZ) -- (YY) -- (XX) ;
\draw (XY) -- (YX) -- (ZZ) ;
\draw (IY) -- (YI) -- (YY) ;
\draw (XI) -- (IX) -- (XX) ;
\end{tikzpicture}
        \caption{Quantum assignment $\alpha_1$ generated by $(G_1,HS(G_1))$. The context illustrated by the double red line is negative.}
        \label{fig:sub3}
    \end{subfigure}
    \caption{Building process from the finite simple graph $G_1$ to the quantum assignment $\alpha_1$.\label{fig:main}}
\end{figure}
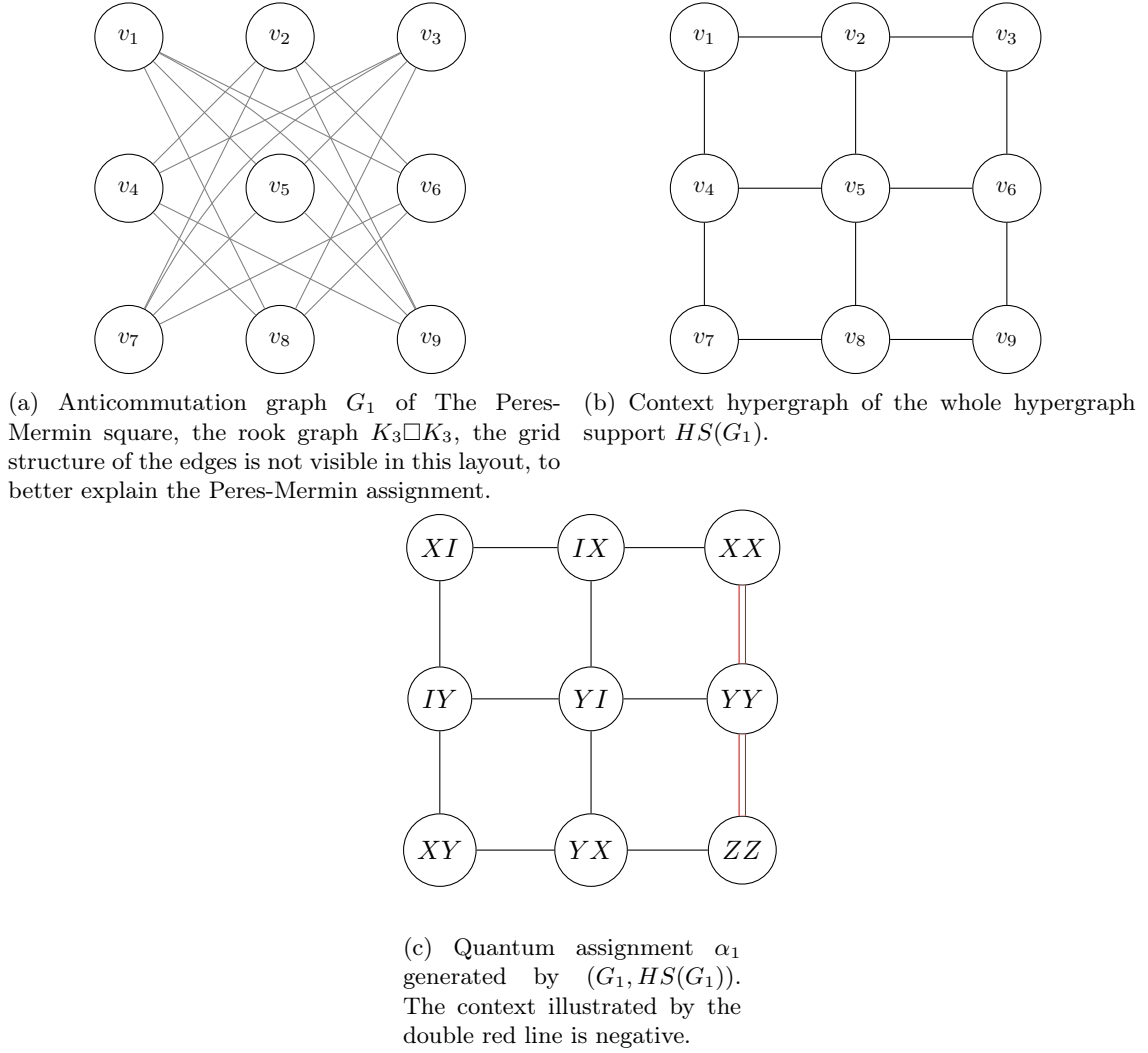

A hypergram $(V,H,G)$ is \emph{(Pauli-)assignable} if it admits at least one Pauli assignment. Following algebraic graph theory~\cite{GR01}, we represent the context hypergraph $(V,H)$ by its incidence matrix $H \in \mathbb{F}_2^{\vert{}H\vert{}\times\vert{}V\vert{}}$ ($H_{k,v}=1$ if $v \in h_k$) and the anticommutation graph $(V,G)$ by its adjacency matrix $G \in \mathbb{F}_2^{\vert{}V\vert{}\times\vert{}V\vert{}}$ ($G_{i,j}=1$ if $\{i,j\} \in G$).

\begin{theorem}[Assignability Conditions]\cite[Theorem 9]{mg25}\label{vecCnsQAprop}
    For any hypergram $(V,H,G)$, the following conditions are equivalent:
    \begin{enumerate}[(i)]\item $(V,H,G)$ admits a Pauli assignment;\item \emph{Matricial condition:} $H \ctimes G = 0$ over $\mathbb{F}_2$, i.,e., every column of $G$ belongs to $\text{ker}(H)$;\item \emph{Graphical condition:} For each vertex $v \in V$, every hyperedge $h \in H$ contains an even number of vertices adjacent to $v$ in $(V,G)$.
    \end{enumerate}
\end{theorem}

\label{labelingSec}When $H \ctimes G = 0$, a Pauli assignment $\alpha : V \rightarrow \{I,X,Y,Z\}^{\otimes n} \setminus \{I^{\otimes n}\}$ can be efficiently constructed in $O(\vert{}V\vert{}^3)$ time by applying standard symplectic reduction steps to $G$ (adapted from~\cite[Theorem 8.10.1]{GR01} and implemented in~\cite[Algorithm 1]{mg25}).The required number of qubits $n$ is directly dictated by the (always even) rank of the anticommutation matrix:

\begin{equation}
    n = \frac{\text{rk}(G)}{2}.
\end{equation}

\begin{definition}[Hypergraph support]\label{hsDef}
Let $G = (V,E)$ be a finite simple reduced graph, The hypergraph support $HS(G)$, is the set $\mathcal{H} = (V,H)$ containing all the contexts that can be formed from $G$ as the set of hyperedges $H$. In other words, it is the set of all independent sets of $G$ in which every vertex $u \in G$ is adjacent with an even number of vertices $v \in h$, for $h \in H$.
\end{definition}

\begin{corollary}
If $(V,H,G)$ is a valid hypergram, then $H \subseteq HS(G)$.
\end{corollary}

Note that a quantum configuration whose hypergram is of the form $(V,HS(G),G)$ is referred to in~\cite{mg25} as a first family hypergram. This family contains most quantum configurations in the literature, such as the Peres-Mermin square (see~\Cref{fig:main}), the Mermin pentagram and the doily.

\subsection{Contextuality degree and noncontextual bound}
\label{contextualitySec}

Let $(V,H,G)$ be a hypergram and $\alpha$ an $n$-qubit Pauli assignment of it. The \emph{Pauli sign function} of $\alpha$ is the map $\mathrm{sgn}_{\alpha} : H \rightarrow \{-1,1\}$ such that $\prod_{v \in h} \alpha(v) = \mathrm{sgn}_{\alpha}(h)\, I^{\otimes n}$ for every hyperedge $h \in H$. A
\emph{classical assignment} is a map $a : V \rightarrow \{-1,+1\}$, and its \emph{classical sign function} is the map $\mathrm{sgn}_{a} : H \rightarrow \{-1,1\}$ given by $\mathrm{sgn}_{a}(h) = \prod_{v \in h} a(v)$.
Sign functions play the same role as the second member of a binary constraint or linear system~\cite{Cleve2014,CLS17} and as the signature of a signed arrangement~\cite{Arkhipov2012}.

A Pauli assignment $\alpha$ is \emph{contextual} if there is no classical assignment $a$ such that $\text{sgn}_{\alpha} = \text{sgn}_{a}$. Quantitatively, following~\cite{DHGMS22,MSGDH24}, the \emph{contextuality degree} $d$ of $\alpha$ is the minimum Hamming distance between its sign function $\text{sgn}_{\alpha}$ and that of any classical assignment $a$:
\begin{equation}
d = \min_{a : V \rightarrow \{-1,+1\}} \left| \{ h \in H \mid \text{sgn}_{\alpha}(h) \neq \text{sgn}_{a}(h) \} \right|.
\label{cdegEq}
\end{equation}

Equivalently, $d$ measures the minimum number of context products on which any noncontextual classical model must disagree with the quantum assignment. 


Let $\alpha$ be a Pauli assignment of an assignable hypergram
$(V,H,G)$, splitting $H$ into its subset $H^{+}$ of positive contexts
and its subset $H^{-}$ of negative ones. Following~\cite{Cab10,TLC22},
any noncontextual hidden-variable model satisfies
\begin{equation}
\chi \;=\; \sum_{h \in H^{+}} \langle h \rangle
         \;-\; \sum_{h \in H^{-}} \langle h \rangle
   \;\leq\; b_{\alpha}(H),
\label{ncIneq}
\end{equation}
where $\langle h \rangle$ is the expectation value of the product of the
observables of the context $h$. This value has already been measured experimentally
in quantum computers~\cite{KH25}.

\begin{definition}[Noncontextual bound]\label{ncBoundDef}
The \emph{noncontextual bound} of the Pauli assignment $\alpha$ is
\begin{equation}
b_{\alpha}(H) \;=\;
  \max_{a : V \to \{-1,+1\}}
  \sum_{h \in H} \mathrm{sgn}_{\alpha}(h)\,\mathrm{sgn}_{a}(h).
\label{ncBoundEq}
\end{equation}
\end{definition}

Quantum mechanics predicts $\langle h \rangle = 1$ for every positive
context and $\langle h \rangle = -1$ for every negative one, whatever the
initial state, so the quantum value of~(\ref{ncIneq}) is
$\chi_{\mathrm{QM}} = |H|$, and the assignment is contextual exactly when
$b_{\alpha}(H) < |H|$. Each hyperedge contributes $+1$
to~(\ref{ncBoundEq}) when it is satisfied and $-1$ when it is not, so
$b_{\alpha}(H) = |H| - 2d$ by~(\ref{cdegEq}), which
is~\cite[Theorem 2]{TLC22}.

\begin{definition}[Tolerated error per context]\label{epsDef}
The \emph{tolerated error per context}~\cite[Eq.~(C3)]{TLC22v2} of an
assignable hypergram of contextuality degree $d$ is
\begin{equation}
\varepsilon \;=\; \frac{\chi_{\mathrm{QM}} - b_{\alpha}(H)}{|H|}
             \;=\; \frac{2d}{|H|}.
\label{epsEq}
\end{equation}
\end{definition}

If the measured value of each context deviates from its ideal quantum
value by at most $\varepsilon$, the inequality~(\ref{ncIneq}) is still
violated; $\varepsilon$ therefore measures the robustness of the
associated state-independent contextuality test to experimental
imperfection, and is the quantity we maximise.

\begin{theorem}~\cite{mg25}~\label{thmSameDeg}
Let $(V,H,G)$ be a hypergram. Then all Pauli assignments of $(V,H,G)$ have the same
contextuality degree, noncontextual bound and tolerated error per context.
\end{theorem}

\section{Methodology}
\label{secMethod}
\subsection{Shrinking the search space}
\label{seCshrinking}

Since~\Cref{thmSameDeg} states that all Pauli assignments of a hypergram share the same contextuality degree, the search for contextuality proofs can be
conducted at the level of hypergrams, without reference to any number of
qubits or to any particular labeling. This leaves the question of how to
enumerate hypergrams $(V,H,G)$, and the order in which their two
components are chosen turns out to matter.

Enumerating context hypergraphs first is impractical. Beyond the growth
of the number of hypergraphs on $|V|$ vertices, each candidate $(V,H)$
must then be completed by every anticommutation graph $G$ contained in
$\mathrm{cplt}(H)$, and only afterwards can the resulting triple be
tested for assignability. The assignability condition, which is the
constraint doing the actual pruning, is thus applied last, on a space
that has already been expanded twice.

We therefore reverse the order and enumerate anticommutation graphs
first. Given a simple reduced graph $(V,G)$, the assignability condition
of~\Cref{vecCnsQAprop} determines which hyperedges may accompany it,
namely those of the hypergraph support $HS(G)$ of \Cref{hsDef}, so that
the admissible context hypergraphs are exactly the subsets of $HS(G)$. 
The constraint is now applied first, and to a single object.

One obstacle remains. $HS(G)$ still has up to
$2^{|HS(G)|} - 1$ non-empty subsets, so enumerating all of them is out of the question.
We settle the matter by taking it entirely, $(V,HS(G),G)$, which we believe has 
the highest chance of providing high $\varepsilon$ values.
Only one hypergram per graph therefore needs to be considered. The
Peres-Mermin square and the Mermin pentagram illustrate the extreme case
of this fact, since removing any one of their contexts leaves a
structure that is no longer contextual. 
As noted in the introduction, the idea of retaining every context compatible with a given
anticommutation graph was first mentioned in~\cite[Section 7.4.5]{muller-thesis} and
is implemented in the \textsc{Qontextium} software~\cite{qo25}.
We note that a technical limitation in the software prevents it from 
checking assignability of hypergrams if they need more than 15 qubits 
(due to the fact that each observable is stored as 32-bit words with 2 bits per qubit, with 2 extra bits for bit-shifting concerns.)

\subsection{Directing the search}
\label{dirsearch}

The number $g_\nu$ of unlabelled simple graphs on $\nu$ vertices grows
super-exponentially, asymptotically as
$g_\nu \sim 2^{\binom{\nu}{2}}/\nu!$~\cite{harary1973graphical}, so an
exhaustive enumeration is out of reach beyond very small sizes and the
search must be directed. Our first pass therefore uses the graphs
collected in the \url{houseofgraphs.org} database~\cite{hog54672}. Being a
curated collection, whose members were recorded because they are of
interest for one reason or another, it is in no sense a uniform sample;
what it offers instead is a set of graphs annotated with their
structural properties, which lets us ask not only which configurations
have a high $\varepsilon$ ratio but which graph properties accompany one.

We ran~\textsc{Qontextium}~\cite{qo25} on the 13864
 graphs of the database
with at most $30$ vertices, that are connected and reduced, retrieved on 26 August 2026. The search
recovered the Peres-Mermin square and the Mermin pentagram, which
validates the pipeline. \Cref{graphplot} shows the ratio of every
configuration found against its number of contexts.

Most graphs with the highest $\varepsilon$ are vertex-transitive, which
suggests the heuristic that a large automorphism group favours a high
ratio. Some caution is needed here, because \url{houseofgraphs.org}
over-represents highly symmetric graphs, so the correlation is in part an
artefact of the sample. It also appears that a high number of contexts does 
also help, probably because each context is an additional constraint to satisfy 
from the point of view of the contextuality degree.

We then made a second pass over the censuses of vertex-transitive
graphs~\cite{GD20}, of which the \url{houseofgraphs.org} database contains only a small portion,
covering 
all such graphs on at most 24 vertices. This
yielded the further configurations marked in \Cref{graphplot}, including
the ones with the highest $\varepsilon$ we have found, which
confirms that the heuristic is worth following rather than merely
descriptive of the first sample.

\section{Results}
\label{secResults}

\begin{figure}[htbp]
    \centering
    \begin{subfigure}[b]{0.99\textwidth}
        \centering
        \includegraphics[width=1.0\textwidth]{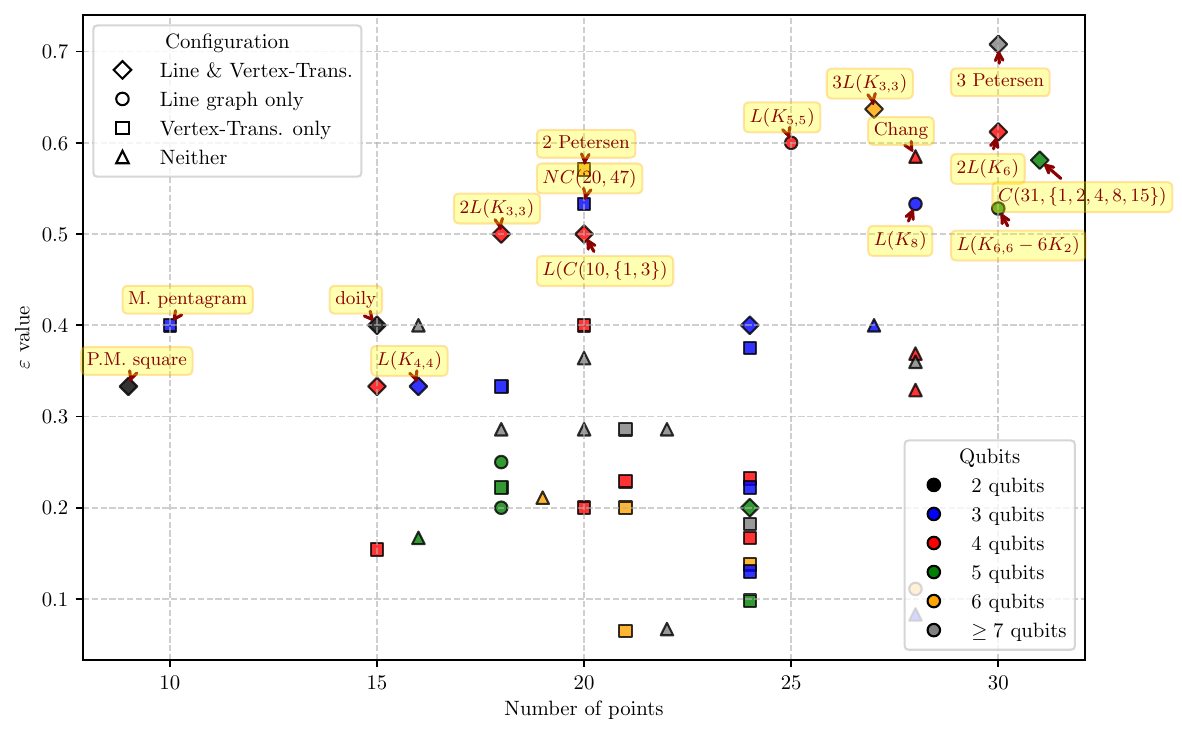}
        \caption{$\varepsilon$ over the number of observables.}
        \label{fig:sub1}
    \end{subfigure}
    \hfill
    \begin{subfigure}[b]{0.99\textwidth}
        \centering
        \includegraphics[width=1.0\textwidth]{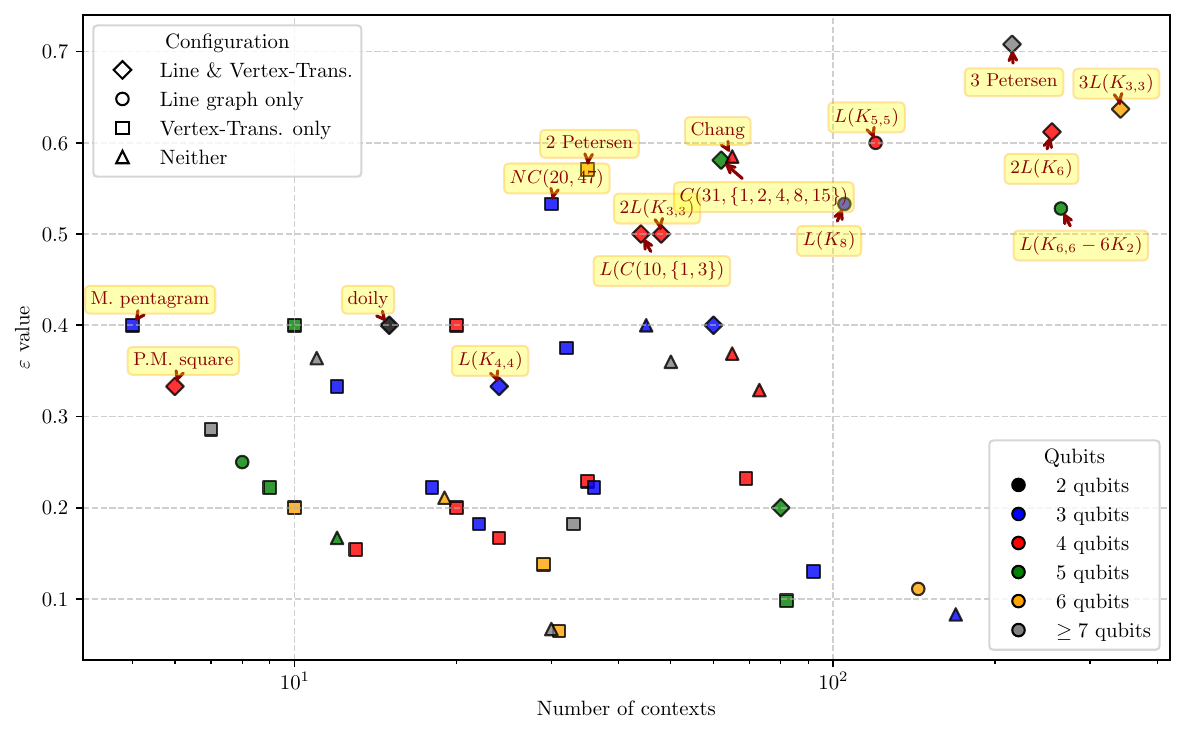}
        \caption{$\varepsilon$ over the number of contexts, the latter being shown in logarithmic scale.}
        \label{fig:sub2}
    \end{subfigure}
    \caption{
        Plots of the $\varepsilon$ values for all the graphs tested, respectively against their number of observables and contexts, with the shapes of the points showing if the configurations are line graphs, and/or vertex-transitive graphs, as well as colors representing their required number of qubits.
        \label{graphplot}
    }
\end{figure}

\begin{table}[htbp]
\centering
\begin{tabular}{lrlll}
\hline
Type & Total & Valid & $n \geq 15$ & Contextual \\
\hline
\url{houseofgraphs.org}$(\leq 30)$  & $13864$ & $2087(15.05\%)$ & $2640(19.04\%)$ & $44(0.32\%)$\\
Vertex-Transitive$(\leq 24)$& $18948$ & $1827(9.64\%)$ & 0 & $34(0.18\%)$\\
\hline
\end{tabular}
\caption{Number of graphs searched, valid hypergrams (i.e. having at least one context) with less than 15 qubits, those with more than 15 qubits that couldn't be checked by \textsc{Qontextium} (before checking assignability) and contextual hypergrams found for each database searched, with their respective presence ratio in the database}
\label{totalTable}
\end{table}

\begin{table}[htbp]
\centering
\begin{tabular}{lrrrrrr}
\hline
Anticommutation graph & HoG id & $|V|$ & $|H|$ & $n$ & $d$ & $\varepsilon$ \\
\hline
Rook $3 \square 3$ / $L(K_{3,3})$ (Peres-Mermin square) & 6607 & 9  & 6  & 2 & 1  & 0.333 \\
Petersen (Mermin pentagram) & 660   & 10 & 5  & 3 & 1  & 0.400 \\
$L(K_6)$ (Doily)            & 54409 & 15 & 15 & 2 & 3  & 0.400 \\
MS3-29~\cite{TLC22v2}       &       & 29 & 33 & 3 & 7  & 0.424 \\
\hline
$NC(20,47)$                 & 54279 & 20 & 30 & 3 & 8  & 0.533 \\
$L(C(10,\{1,3\}))$            & 52168 & 20 & 44 & 4 & 11 & 0.500 \\
Rook $5 \square 5$ / $L(K_{5,5})$                & 51213 & 25 & 120& 4 & 36 & 0.600 \\
Chang                       & 6626  & 28 & 65 & 4 & 19 & 0.585 \\
$L(K_8)$                    & 54672 & 28 & 105& 3 & 28 & 0.533\\
$L(K_{6,6} - 6K_2)$         & 54732 & 30 & 265& 5 & 70 & 0.528 \\

\hline
$2 L(K_{3,3})$&&18& 48 & 4 & 12 & 0.500 \\
2 Petersen  & & 20 & 35 & 6 & 10 & 0.571 \\

\hline
Rook $4 \square 4$ / $L(K_{4,4})$ & 30317 & 16 & 24& 3 & 4 & 0.333\\

$C(31,\{1,2,4,8,15\})$      & 51217 & 31 & 62 & 5 & 18 & 0.581 \\
$2 L(K_6)$    & & 30 & 255& 4 & 78  & 0.612\\
$3 L(K_{3,3})$& & 27 & 342& 6 & 109 & 0.637\\
3 Petersen    & & 30 & 215& 9 & 76  & 0.707\\

Rook $6 \square 6$ / $L(K_{6,6})$ & 55009 & 36 & 720& 5 & $\leq$ 280 & $\leq$ 0.778\\
Triangular 10 / $L(K_{10})$ & 55626 & 45 & 945& 4 & $\leq$ 381 & $\leq$ 0.806\\

\hline
\end{tabular}
\caption{Comparison of contextual configurations. The first block represents already known contextual configurations and the configuration yielding the previous record~\cite{TLC22v2}, along with their anticommutation graph when known. The second block configurations with the highest contextuality ratio found during the \url{houseofgraphs.org} research phase, and the third one represents the vertex-transitive phase of the search (those which weren't present in the former one). The last block represents values for the extended search on hypergrams of circulant graphs, line graphs and disjoint graphs following~\Cref{linegraphthm}.
Values of $d$ marked with $\leq$ are upper bounds
obtained by the heuristic solver~\cite{MSGHK24}, the others are exact.}
\label{resultsTable}
\end{table}

The statistics about the graphs studied in the \url{houseofgraphs.org} database and the vertex-transitive ones are shown in~\Cref{totalTable}, where we can observe that few graphs lead to assignable hypergrams. We attribute this to the parity condition of \Cref{hsDef} being a demanding requirement on an independent set. Moreover, of the 13864 graphs searched, only $2087/11224 = 18.6\%$ yielded a non-empty hypergraph support, at that rate the 2640 graphs beyond Qontextium's reach due to needing more than 15 qubits would contribute on the order of five hundred further configurations. However, high-rank graphs leave a smaller kernel, meaning there is less room for contexts to appear, so we expect the true figure to be lower still. A strong argument towards this hypothesis is that the highest number of qubits for a contextual configuration found during the search was 10. 
The most striking result of this research is that we are now able to find configurations whose $\varepsilon$ is above $0.424$ and even $0.5$, the previous unpublished record set by the $44$-configuration in~\cite[Section 7.5]{muller-thesis}, that we identify here as $L(C(10,\{1,3\}))$.
As stated in~\Cref{dirsearch}, the data seems to back the heuristic of a high number of contexts increasing $\varepsilon$. Additionally, most high $\varepsilon$ graphs are line graphs of other graphs (although this is not the case for all of them, the Petersen graph yielding the Mermin Pentagram being the simplest example.) This observation led to the discovery of~\Cref{linegraphthm}. The results are plotted in~\Cref{graphplot} and all the configurations with an $\varepsilon$ greater or equal than 0.5 along with some interesting others are shown in~\Cref{resultsTable}.

Before going through the next two families of graphs that seem the most interesting, we mention here that the presence of a lot of circulant
graphs motivated us to look for more such graphs, with up to 33 vertices. Only one configuration came out of this search but its $\varepsilon = 0.581$
makes it worth being included in~\Cref{resultsTable}. 

\subsection{Line graphs}

\begin{theorem}\label{linegraphthm}
Let $H = (V, E)$ be a simple graph, and let $G = L(H)$ be its line graph with adjacency matrix $A$. Every perfect matching $M$ of $H$ forms a context in $G$; that is, $M$ is an independent set in $G$ satisfying
\[
\sum_{e \in M} A_e \equiv \mathbf{0} \pmod 2.
\]
\end{theorem}

\begin{proof}
To establish that $M$ is a context in $L(H)$, two conditions must hold:
\begin{enumerate}
    \item \textbf{Independent Set:} Vertices in $L(H)$ correspond to edges in $H$, with adjacency defined by sharing an endpoint. Because $M$ is a matching in $H$, no two edges in $M$ share an endpoint. Consequently, $M$ forms an independent set in $L(H)$.
    \item \textbf{Kernel Condition over $\mathbb{F}_2$:} We require that for every vertex $f \in V(L(H)) = E(H)$, the column sum $S_f = \sum_{e \in M} A_{e,f} \equiv 0 \pmod 2$.
    \begin{itemize}
        \item \textbf{Case 1 ($f \in M$):} Because $M$ is independent in $L(H)$, $f$ is not adjacent to any vertex $e \in M$. Thus, $A_{e,f} = 0$ for all $e \in M$, yielding $S_f = 0 \equiv 0 \pmod 2$.
        \item \textbf{Case 2 ($f \notin M$):} Let $f = \{u, v\} \in E(H) \setminus M$. Because $M$ is a perfect matching, endpoints $u$ and $v$ are each incident to a unique edge in $M$, denoted $e_u$ and $e_v$ respectively. Since $H$ is a simple graph and $f \notin M$, $e_u$ and $e_v$ are distinct from each other and from $f$. Hence, $f$ shares an endpoint with exactly two distinct edges in $M$ ($e_u$ at vertex $u$, and $e_v$ at $v$). This implies $A_{e_u,f} = 1$, $A_{e_v,f} = 1$, and $A_{e,f} = 0$ for all $e \in M \setminus \{e_u, e_v\}$.
    \end{itemize}
    Summing over all elements in $M$ gives:
    \[
    S_f = A_{e_u,f} + A_{e_v,f} = 1 + 1 = 2 \equiv 0 \pmod 2.
    \]
\end{enumerate}
Because $S_f \equiv 0 \pmod 2$ holds for every $f \in V(L(H))$, we have $\sum_{e \in M} A_e \equiv \mathbf{0} \pmod 2$. Thus, $M$ is a context in $L(H)$.
\end{proof}

\Cref{linegraphthm} gives a direct way to produce configurations with
many contexts since the contexts supported by $L(G)$ contain
the perfect matchings of $G$, it suffices to choose base graphs
$G$ with many of them. We therefore concentrated a final phase of
the search on the following two families: the triangular
graphs $L(K_{2k})$ and the rook's graphs $L(K_{m,m})$, whose base
graphs maximise the number of perfect matchings among graphs on a given
number of vertices, and whose first assignable hypergram representatives are 
respectively the doily at $k=3$ and the Peres-Mermin square at $m=3$.

The rook's graph $K_4 \square K_4$, on $16$ observables and
$24$ contexts, has contextuality degree $4$ and thus
$\varepsilon = 1/3$, the same value as the Peres-Mermin square
$K_3 \square K_3$, whereas $K_5 \square K_5$ reaches $\varepsilon = 0.6$.
The ratio is therefore not monotone along the family, and the growth of
the number of perfect matchings alone does not necessarily account for the high
values $\varepsilon$ observed. Beyond these sizes, the exact computation of the degree
became the limiting factor. Indeed for $K_6 \square K_6$ ($36$ observables,
$720$ contexts) and for $L(K_{10})$ ($45$ observables, $945$ contexts),
the SAT solver did not terminate within our computational budget. The
heuristic method of~\cite{MSGHK24} still returns values for these two
configurations. Although they are only upper bounds on the degree, these 
values have a chance of being tight since on every configuration of
\Cref{resultsTable} for which both methods terminated they returned the
same value.

\subsection{Disjoint graphs}

\label{sec:disjoint}
 
Looking at the configurations built from the vertex-transitive database, only two 
configurations had $\varepsilon > 0.5$. A study of these graphs unveiled the fact that 
they are a disjoint union of two $L(K_{3,3})$ for one of them, and two $\mathrm{Petersen}$
for the other one, or in other words two copies of the anticommutation graphs of the Peres-Mermin 
square and of the Mermin pentagram respectively. Contexts
combine freely across components, since for a disjoint union independence and the parity
condition of~\Cref{hsDef} are checked componentwise, giving
\[
|HS(G_1 \sqcup \cdots \sqcup G_k)| = \prod_{i=1}^{k}\bigl(|HS(G_i)|+1\bigr) - 1 ,
\]
multiplicative where the number of observables is only additive: $7^2-1=48$ and $6^2-1=35$
contexts, for $\varepsilon = 0.5$ and $0.571$ against $1/3$ and $2/5$ for a single copy.

For this reason we looked for more of them. The \url{houseofgraphs.org} database holds few
disconnected graphs on at most 30 vertices, and that pass returned nothing new beyond one
configuration of ratio $\varepsilon = 0.333$. Extrapolating the classical examples proved
far more effective: $3L(K_{3,3})$ reaches $\varepsilon = 0.637$ and $2L(K_6)$ reaches
$0.612$, while $3\,\mathrm{Petersen}$, with 30 observables, 215 contexts and $d = 76$,
attains $\varepsilon = 0.707$ — our record, and the largest instance the SAT solver still
handled within our computational budget.

\section{Finite Geometric Rendering of Some Configurations}
\label{secGeo}

To get further insights into the nature of contextual configurations found, it is worth providing
finite geometric representations of some of their most interesting representatives. This brief
finite geometric dive will illustrate that contextual sets are intricately linked with the structure
of specific symplectic polar spaces and, in particular, with some distinguished objects (like Fano
planes, ovoids and/or hyperbolic quadrics) living in them. 

\subsection{$NC(20,47)$ Configuration}
The anticommutation graph $\mathrm{NC}(20,47)$ is a vertex-transitive graph on $20$
vertices which is not a Cayley graph, its label 47 being its number assigned in the~\url{houseofgraphs.org} database as one of the 82 20-vertex graphs with both properties.

This three-qubit configuration has 20 observables
and 30 contexts of size four each. As a context represents an affine plane of order two, by adding a `line at infinity'
we can extend/complete each of them to a Fano plane. In doing so, we find that each such line is shared
by two planes, so we have only 15 lines. These 15 lines form a (quadratic) doily that lies
on a unique hyperbolic quadric of $W_3$ whose index is $YYI$.
As there are 1008 doilies lying on quadrics in $W_3$, and each quadric host just one configuration,
we get altogether 1008 configurations of this type. Going the other way round, pick up one hyperbolic quadric in
$W_3$ and one of 28 doilies lying on it. There are three planes through each line of the doily, so 45 planes altogether.
Then one disregards those 15 planes (one per each line of the doily) that pass through
the same point. From the remaining 30 planes (forming two systems of generators on the selected quadric)
one removes the 15 lines (as well as 15 points) of the doily and  gets the configuration isomorphic
to our NC(20,47)-one.

Moreover, from the above construction it follows that each such configuration contains
12 Mermin pentagrams. They come in pairs, each pair being associated with a particular spread of lines of the doily
\cite{LSz17}. An example is shown in Figure \ref{fig:penta}. \\

\begin{figure}[t]
\centerline{\includegraphics[width=12truecm,clip=]{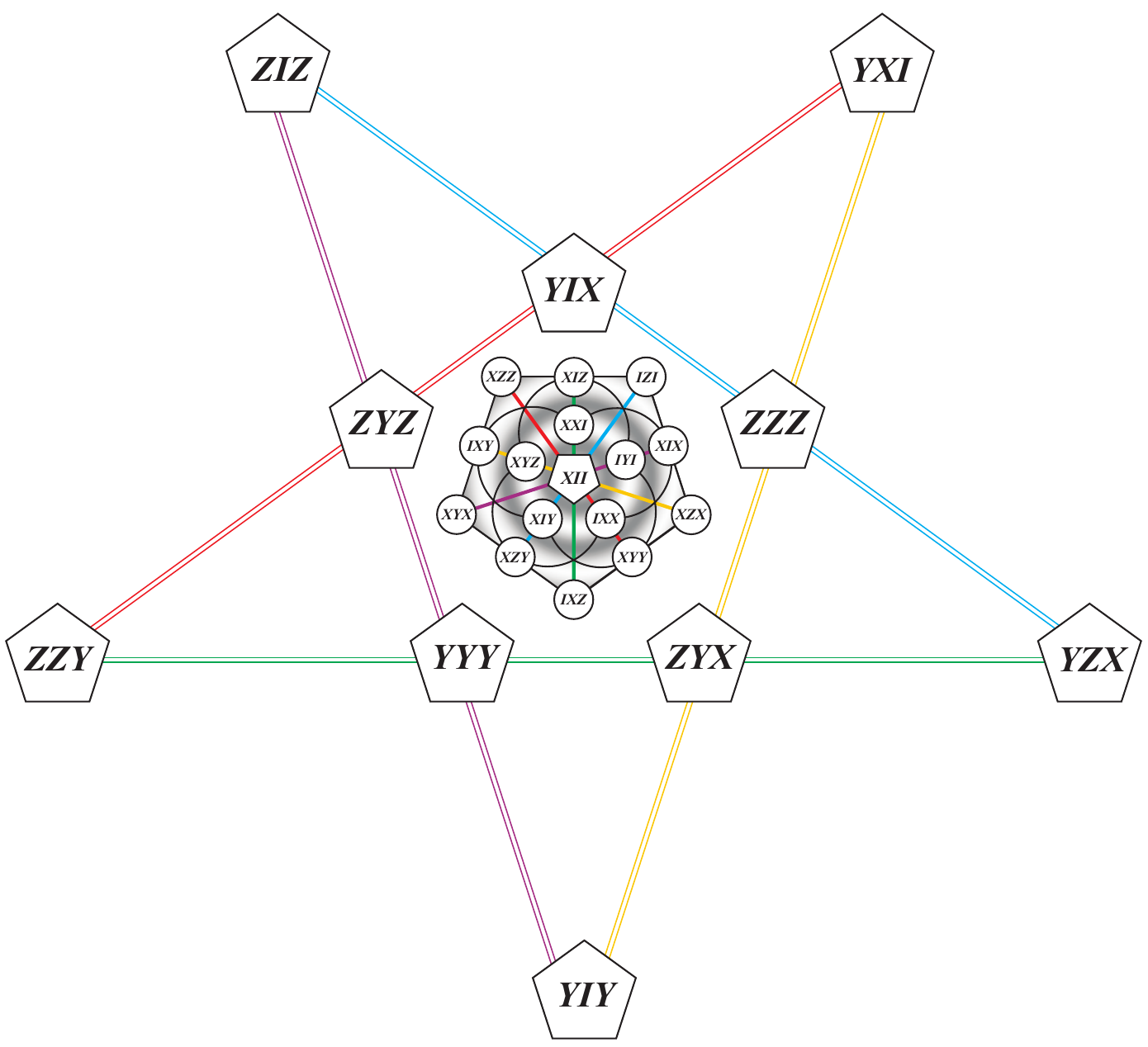}}
\vspace*{.2cm}
\caption{A doily (middle) in $W_3$ lying on a certain hyperbolic quadric, one of its spread of lines (five colored pairwise disjoint lines) and one
of the two associated Mermin pentagrams. An edge of the pentagram of a given color forms with the line of the doily that is of the same color a Fano plane. \label{fig:penta}}
\end{figure}

\subsection{Triangular 8 ($L(K_8)$) Configuration}
This is again a three-qubit configuration, having now 28 observables and 105 four-element
contexts. It contains a {\it unique} Conwell heptad, namely 
$IIX$, $IIZ$, $IXY$, $IZY$, $XYY$, $ZYY$ and $YYY$. This heptad defines a {\it unique}
hyperbolic quadric of index $YIY$, and the 28 observables of our
configuration are nothing but the 28 exterior points to this quadric. We also have 15 contexts 
through each point of the heptad, which originate from 15 planes passing
through each point; the $(7 \times 15 =)$105 `lines at infinity' we get by completing each context to a Fano plane all lie on our
hyperbolic quadric. As there are 36 hyperbolic quadrics in $W_3$ and each is associated with eight Conwell heptads,
we will have altogether $(36 \times 8=)$ 288 configurations of this type.\\

\subsection{$L(K_{5,5})$ Configuration}
This four-qubit configuration of 25 observables and
120 five-element contexts is highly symmetric. Let us call a set
of five pairwise anticommuting observables whose product is different from
$IIII$ a pentad. Then our set of 25 observables contains exactly ten such pentads;
moreover, it can be partitioned into five pentads in two different ways as illustrated
in Figure \ref{fig:25-120}. \\

\begin{figure}[t]
\centerline{
    \begin{tikzpicture}[
  every node/.style={
    circle, draw, fill=white,
    inner sep=0pt, minimum size=11mm,
    font=\itshape\footnotesize,
    scale=0.8
  },
  line width=0.5pt,
  scale=0.8
]

\foreach \lab [count=\k from 0] in {IIIX,IIIZ,IIXY,IIZY,IXYY}
  \node (n1-\k) at (90-72*\k:2.05) {\lab};
\foreach \lab [count=\k from 0] in {IZYZ,IZYX,IZZI,IZXI,IYII}
  \node (n2-\k) at (90-72*\k:3.35) {\lab};
\foreach \lab [count=\k from 0] in {XYYZ,XYYX,XYZI,XYXI,XZII}
  \node (n3-\k) at (90-72*\k:4.65) {\lab};
\foreach \lab [count=\k from 0] in {YYYZ,YYYX,YYZI,YYXI,YZII}
  \node (n4-\k) at (90-72*\k:5.95) {\lab};
\foreach \lab [count=\k from 0] in {ZYYZ,ZYYX,ZYZI,ZYXI,ZZII}
  \node (n5-\k) at (90-72*\k:7.25) {\lab};

\begin{scope}[on background layer]
  \foreach \lay in {1,...,5}{%
    \foreach \k in {0,...,4}{%
      \pgfmathtruncatemacro{\nextk}{mod(\k+1,5)}
      \draw (n\lay-\k) -- (n\lay-\nextk);
    }
  }
  \foreach \k in {0,...,4}{%
    \foreach \lay in {1,...,4}{%
      \pgfmathtruncatemacro{\nextlay}{\lay+1}
      \draw (n\lay-\k) -- (n\nextlay-\k);
    }
  }
\end{scope}

\end{tikzpicture}
}
\vspace*{.2cm}
\caption{A portrayal of the fact that 25 observables of our $L(K_{5,5})$-configuration can be partitioned
into five pentads in two different ways; the five pentads in one set are
represented by pentagons, the ones in the other set by lines. Each context thus picks up one observable from each pentad.\label{fig:25-120}}
\end{figure}
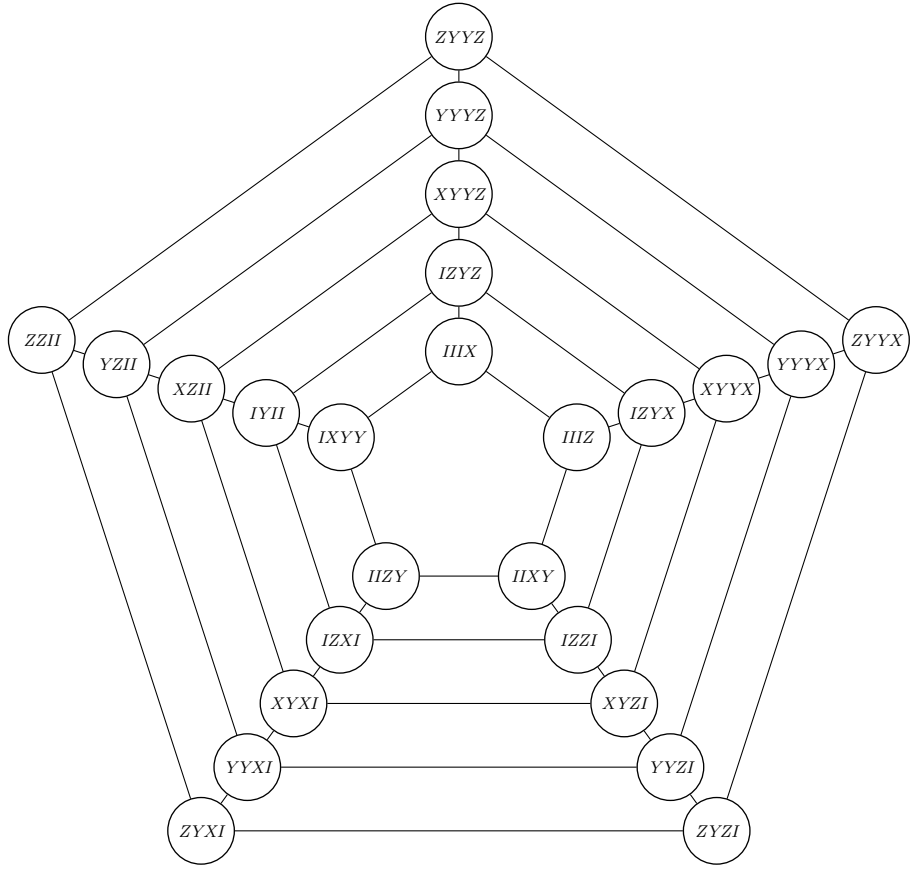

\subsection{Chang Configuration}
The anticommutation graph is one of the Chang graphs i.e. the three strongly regular graphs
whose parameters $(28,12,6,4)$, are the same as $L(K_8)$.
This four-qubit configuration featuring 28 observables and 65 four-element contexts is also remarkable. 
It has 18 observables each of which is on ten contexts, whilst each of the remaining 10 observables belongs
only to eight contexts. Moreover, these ten observables themselves form a four-qubit Mermin pentagram whose five edges
are
\begin{center}
$IIZY-IZZY-XIII-XZII$,

$IIZY-ZXZY-IXII-ZIII$,

$XIYY-IXYY-IXII-XIII$,

$XIYY-ZXZY-ZYXI-XZII$,

$IXYY-IZZY-ZYXI-ZIII$. 
\end{center}

We further observe that 18 observables of order ten can be split into six
triples such that each triple defines a plane in the ambient
PG(7,2) and the corresponding six planes share a line, namely
$IIXX-IIXZ-IIIY$. The shape of the whole configuration is schematically depicted in
Figure \ref{fig:28-65}.\\

\begin{figure}[pth!]
\centerline{\includegraphics[width=1.2\textwidth,clip=]{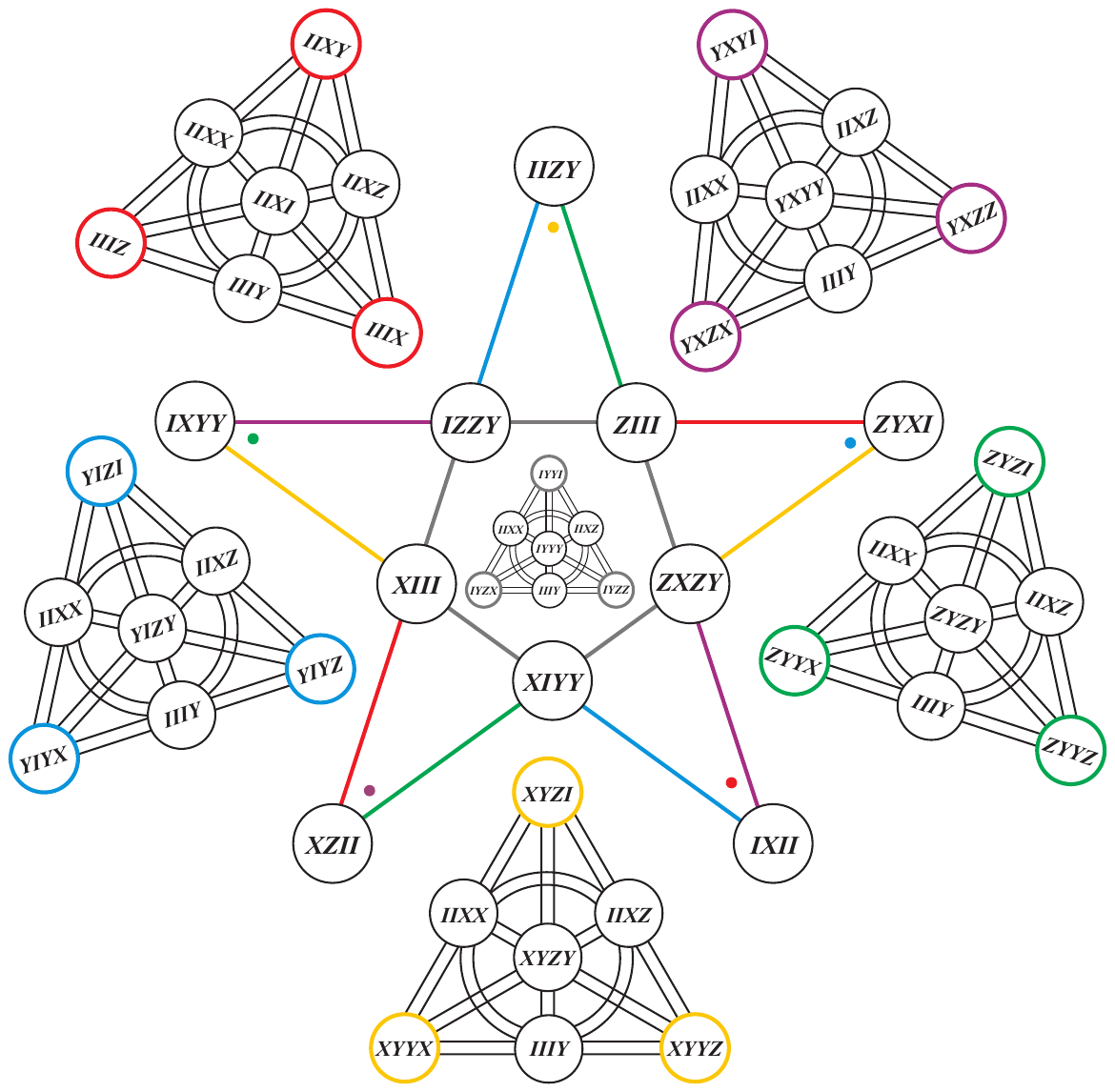}}
\vspace*{-6cm}
\caption{The composition of the Chang configuration.  The central (Mermin) pentagram features ten points of
order eight and the six concurrent Fano planes (the line of concurrence is always
represented by a circle) comprise 18 observables of order ten (colored).
Each observable of the pentagram is the meet of three colors and it commutes
with each observable in the Fano planes of the corresponding colors.\label{fig:28-65}}
\end{figure}

\subsection{$C(31,\{1,2,4,8,15\})$  Configuration}
A configuration living in $W_5$ that contains 31 observables and 62 six-element contexts.
It is also highly symmetric as each observable belongs to the same number of contexts -- 12.
Given any observable, there are ten observables anticommuting with it. Any pair of anticommuting
observables share three observables that anticommute with both of them. Remarkably, these three
observables are mutually commuting and so define a unique Fano plane in $W_5$; hence we have
a distinguished set of $(31 \times 10/2 =)$ 155 planes associated with this configuration.
\subsection{Triangular 10 ($L(K_{10})$) Configuration}
A $45-945$ configuration that is
basically a four-qubit analog of the $L(K_8)$-one. Unlike the three-qubit case, we have now as many as ten
distinguished ovoids (that are now of size nine), two per each observable; one of them is $IIIX$, $IIIZ$, $IIXY$, $IIZY$, $IXYY$, $IZYY$,
$XYYY$, $ZYYY$ and $YYYY$. Furthermore, whereas a three-qubit Conwell heptad is disjoint from its associated quadric, this ovoid 
lies {\it fully} on a certain hyperbolic quadric, namely that of index $YIYI$.
Through each of these nine observables there pass 135 PG(3,2)s of $W_4$,
of which 30 are located on the quadric itself. Each of the remaining
105 three-spaces hosts a single five-element context, so one gets 105 contexts
per each ovoid's observable and, hence, $(105 \times 9 =$) 945 contexts altogether.

\subsection{Configurations from two disjoint graphs}
The two disjoint configurations differ from any previously discussed one by the fact
that their contexts are of {\it two} different sizes. The 2 $ L(K_{3,3})$-configuration, which 
lives in $W_4$, has 18 observables that are located on 12 contexts of size three (i.\,e., lines)
and 36 ones of size six (each representing a pair of disjoint lines). Remarkably, the 12 three-element 
contexts form nothing but a pair of disjoint four-qubit {\it Peres-Mermin squares} and each six-element 
context picks up one line from either square. Similarly, the 2 Petersen-configuration, which belongs to $W_6$,
features 20 observables occupying ten contexts of size four (i.\,e., affine planes of order two) and 25 ones
of size eight (a pair of disjoint affine planes). Analogously to the preceding case, the ten four-element 
contexts are split into two disjoint equally-sized sets, each representing a six-qubit {\it Mermin pentagram}, 
and any of the remaining eight-element contexts has equal, four-element share with either pentagram.

\subsection{Configurations from three disjoint graphs}
These configurations exhibit the largest $\varepsilon$ found so far.
The 3 $ L(K_{3,3})$ configuration of this kind belongs to $W_6$ 
and has 27 observables and 
342 contexts. The 27 observables are located on three mutually
disjoint Peres-Mermin squares such that each observable in one
square commutes with any observable in the other two squares.
The 342 contexts split into 18 contexts of size three (each representing
a line of a square),
108 contexts of size six (each representing a pair of disjoint lines from
two different squares) and 216 contexts of size nine (each picking a line
from each of the three squares). Its degree of contextuality $d=109$, 
which yields $\varepsilon = 0.637$.

The 3 Petersen configuration lives in $W_9$.
It has 30 observables and 215 contexts. The 30 observables are equally distributed into
three pairwise skew Mermin pentagrams such that each observable in one pentagram
commutes with any other observable of the remaining two pentagrams.
Out of its 215 contexts, 30 are of size four (each representing an edge of a pentagram),
75 of size eight (each comprising a pair of edges from two different pentagrams) and
125 of size twelve (each featuring an edge from each of the three pentagrams). Having
$d=76$, its $\varepsilon = 0.707$, the largest confirmed one.

\section{Discussion and Outlooks}
\label{seCdiscussion}

Because the contextuality degree is an invariant of the underlying
hypergram~\cite{mg25}, the search for observable-based Kochen-Specker
proofs can be conducted entirely at the level of graphs, with no
reference to qubits or observables. We have made that search concrete by systematically
taking into account all
the contexts compatible with an anticommutation graph $G$ forming its
hypergraph support $HS(G)$. Narrowing the search to such canonical
candidate configurations seems like a good approach to 
maximise the contextuality degree, and allowed us to focus our research 
on graphs alone.

Applied to the \url{houseofgraphs.org} database and to censuses of vertex-transitive graphs, this research recovers the Peres-Mermin square, the Mermin pentagram and the doily, and produces configurations reaching $\varepsilon = 0.707$, for disjoint unions and $\varepsilon=0.6$ for connected anticommutation graphs. The previous published record was $0.424$~\cite{TLC22v2}, although, as explained in the introduction, their search was restricted to magic sets. Two structural features stand out among the best of them. The first is being a line graph, explained by~\Cref{linegraphthm}: perfect matchings of $G$ are a source of contexts for $L(G)$. That theorem identifies the Peres-Mermin square and the doily as the first members of two infinite families, but it does not predict the degree: $K_4 \square K_4$ has four times as many contexts as $K_3 \square K_3$ for the same ratio $1/3$, while $K_5 \square K_5$ reaches 0.6. Many contexts are therefore necessary for a high $\varepsilon$, but plainly not sufficient. The second feature, and the one behind our record, is disconnectedness: independence and the parity condition of~\Cref{hsDef} being componentwise, the number of contexts of a disjoint union is multiplicative while the number of observables is only additive. Three copies of the Petersen graph, i.e. three Mermin pentagrams, thus attain $\varepsilon = 0.707$, ahead of 3 $ L(K_{3,3})$ at $0.637$ and 2 $ L(K_6)$ at $0.612$ -- stacking copies of the smallest classical examples beats every connected configuration we found, and whether $\varepsilon$ tends to a limit or to 1 as the number of copies grows is open.
Furthermore, it could entirely be possible that a subset of $HS(G)$ could
lead to a higher $\varepsilon$ than the whole support itself.
That gap is the first perspective. A useful next step could be to find an 
invariant that could predict the ratio $\varepsilon$ better than only counting
the number of contexts.

The second perspective is the size barrier. Our exhaustive computations
stop where the SAT solver does, around a thousand contexts, and the
graphs whose structure most suggests a high $\varepsilon$ are exactly those which
are out of reach: $K_6 \square K_6$, $L(K_{10})$ and unions of more than three graphs already are.
Enlarging the search to bigger graphs is therefore contingent on
computing degrees differently. Unable to provide exact values, the
heuristic approach of~\cite{MSGHK24} remains usable and agreed with the solver on
every instance where both terminated, but drawing conclusions from that remains 
risky without knowing if these bounds are tight. 

A last perspective concerns the exploration itself. The search we
describe is a fixed pipeline applied to fixed graph databases, and its
main cost is human: choosing which families to try next, noticing that
the good configurations were line graphs, and conjecturing the statement
of \Cref{linegraphthm}. Large language models are increasingly used to
automate such loops, generating candidate constructions, testing them
against a verifier and refining them accordingly. The setting is
favourable here, because every step of our pipeline is machine-checkable.
An automated loop could thus propose graph
families and conjectured supports and have them confirmed or refuted
without human intervention, the human task shifting to identifying 
and checking the most useful statements.
\section*{Acknowledgments} 

We would like to thank Petr Pracna for his contribution in the production of~\Cref{fig:penta} and~\ref{fig:28-65}.

\section*{Financial support}

This work has received funding from ChistEra-2023/05/Y/ST2/00005 under the project 
Modern Device Independent Cryptography (MoDIC). This work was also supported in part 
by the Slovak VEGA grant agency, project number 2/0043/24.



\section*{Use of generative AI}

Generative artificial intelligence tools were used during the preparation
of this article, as follows. Google Gemini 3.5 Flash was used to explore possible
regularities among the configurations returned by the search, 
and discovered~\Cref{linegraphthm}, which was reformulated and proved by the
authors.
Anthropic Claude Opus 5 was used during the writing process of this article.
The computational results of \Cref{resultsTable} and \Cref{graphplot} were produced by the
authors own software and no AI tool took part in their computation. The
authors take full responsibility for the content of this article.

\section*{Data Availability Statement} 

The data that support the findings of this study are openly available in
\url{https://doi.org/10.5281/zenodo.22790412}

\printbibliography

\end{document}